\documentclass[11pt]{article}

\usepackage[hyperref]{jmlr2e}
\usepackage{times}

\usepackage{amsmath,amssymb,bm,xspace,graphicx}
\usepackage{subcaption}
\usepackage{algorithm,algorithmicx}
\usepackage[noend]{algpseudocode}
\usepackage{graphicx}
\usepackage{tikz}
\usetikzlibrary{positioning, arrows, arrows.meta}

\newcommand{\caC}{\mathcal{C}}
\newcommand{\caD}{\mathcal{D}}

\newcommand{\caV}{\mathcal{V}}

\newcommand{\bbN}{\mathbb{N}}
\newcommand{\bbR}{\mathbb{R}}
\newcommand{\bmb}{{\bm b}}

\newcommand{\bme}{{\bm e}}

\newcommand{\bmp}{{\bm p}}

\newcommand{\bmw}{{\bm w}}
\newcommand{\bmx}{{\bm x}}
\newcommand{\bmy}{{\bm y}}
\newcommand{\bmz}{{\bm z}}
\newcommand{\bmone}{{\bm 1}}
\newcommand{\bmzero}{{\bm 0}}

\newcommand{\bmeta}{{\bm \eta}}
\newcommand{\bmphi}{{\bm \varphi}}

\newcommand{\set}[1]{\{#1\}}
\newcommand{\bigset}[1]{\bigl\{#1\bigr\}}
\newcommand{\Bigset}[1]{\Bigl\{#1\Bigr\}}
\newcommand{\abs}[1]{\lvert #1\rvert}
\newcommand{\norm}[1]{\lVert #1\rVert}
\newcommand{\inprod}[1]{\langle #1\rangle}

\newcommand{\argmax}{\mathop{\mathrm{argmax}}}

\newcommand{\Lovasz}{Lov{\'a}sz\xspace}

\allowdisplaybreaks

\ShortHeadings{Polynomial-Time Algorithms for Submodular Laplacian Systems}{Fujii, Soma, and Yoshida}
\firstpageno{1}
\begin{document}
\title{Polynomial-Time Algorithms for Submodular Laplacian Systems }
\author{%
    \name Kaito Fujii \email kaito\_fujii@mist.i.u-tokyo.ac.jp 
       % \addr Graduate School of Information Science and Technology \\
        % The University of Tokyo,\\
        % 7-3-1 Hongo, Bunkyo-ku, Tokyo 113-8656, Japan
       \AND
    \name Tasuku Soma \email tasuku\_soma@mist.i.u-tokyo.ac.jp \\
       \addr Graduate School of Information Science and Technology \\
        The University of Tokyo,\\
        7-3-1 Hongo, Bunkyo-ku, Tokyo 113-8656, Japan
       \AND
       \name Yuichi Yoshida \email yyoshida@nii.ac.jp \\
       \addr National Institute of Informatics, \\
       2-1-2 Hitotsubashi, Chiyoda-ku, Tokyo 101-8430, Japan
       %Preferred Infrastructure, Inc.\\
       %1-6-1 Otemachi, Chiyoda-ku, Tokyo 100-0004, Japan
       }
\editor{}

\maketitle
\begin{abstract}
  Let $G=(V,E)$ be an undirected graph, $L_G\in \mathbb{R}^{V \times V}$ be the associated Laplacian matrix, and ${\bm b} \in \mathbb{R}^V$ be a vector.
  Solving the Laplacian system $L_G {\bm x} = {\bm b}$ has numerous applications in theoretical computer science, machine learning, and network analysis.
  Recently, the notion of the Laplacian operator $L_F:\mathbb{R}^V \to 2^{\mathbb{R}^V}$ for a submodular transformation $F:2^V \to \mathbb{R}_+^E$ was introduced, which can handle undirected graphs, directed graphs, hypergraphs, and joint distributions in a unified manner.
  In this study, we show that the submodular Laplacian system $L_F({\bm x}) \ni {\bm b}$ can be solved in polynomial time.
  Furthermore, we also prove that even when the submodular Laplacian system has no solution, we can solve its regression form in polynomial time.
  Finally, we discuss potential applications of submodular Laplacian systems in machine learning and network analysis.
\end{abstract}

%!TEX root=./main.tex

\section{Introduction}

In spectral graph theory, the \emph{Laplacian matrix} (or simply \emph{Laplacian}) $L_G = D_G-A_G$ associated with an undirected graph $G=(V,E)$ is an important object, where $D_G \in \bbR^{V \times V}$ is a diagonal matrix with the $(v,v)$-th element equal to the degree of $v \in V$ and $A_G \in \bbR^{V \times V}$  is the adjacency matrix of $G$.
Using the Laplacian $L_G$ of an undirected graph $G$, one can extract various information regarding $G$, such as commuting time of random walks, maximum cut, and diameter of the graph.
Further, a cornerstone result in spectral graph theory is Cheeger's inequality~\citep{Alon:1986gz,Alon:1985jg}, which associates the community structure of $G$ with the second smallest eigenvalue of $L_G$.
See~\citet{chung1997spectral} for a survey on this area.

An important problem considering Laplacians for undirected graphs is solving the \emph{Laplacian system} $L_G \bmx = \bmb$ for a graph $G=(V,E)$ and a vector $\bmb\in \bbR^V$;
in terms of electrical circuits, this problem can be interpreted as follows:
We regard each edge $e\in E$ as a resistance of $1\Omega$ and each vertex $v \in V$ as a joint connecting these resistances.
Then, the solution $\bmx$ provides the electric potential at the vertices when a current of $\bmb(v)$A flows through each $v \in V$.
Based on this interpretation, it is evident that the Laplacian system has a solution only when $\sum_{v \in V}\bmb(v) = 0$, that is, the amount of inflow is equal to that of outflow.
%Then, $\bmx$
%the effective resistance $R_{uv}$ is equal to the voltage difference between $u$ and $v$ when we flow a current of $1A$ from $u$ to $v$.

Solving Laplacian systems has numerous applications such as in simulating random walks~\citep{Cohen:2016gm}, generating spanning trees~\citep{Kelner:2009ca}, constructing sparsifiers~\citep{Spielman:2011ir}, faster interior point methods~\citep{Daitch:2008je}, semi-supervised learning~\citep{Joachims:2003vr,Zhou:2003uv,Zhu:2003tb}, and network analysis~\citep{Brandes:2005bk,Hayashi:2016ub,Mavroforakis:2015to,Newman:2005em}.
For more applications of solving Laplacian systems, refer  to~\cite{Vishnoi:2013fz}.

Although the concept of Laplacians for undirected graphs was first introduced in the 1980s, it is only recently that the notions of Laplacian operators for directed graphs~\citep{Yoshida:2016ig} and hypergraphs~\citep{Louis:2015tg} have been  proposed and corresponding Cheeger's inequalities have been obtained.\footnote{Precisely speaking, several Laplacian \emph{matrices} have been proposed for directed graphs~\citep{Chung:2007bj,Li:2012tk,Zhou:2005bua}; however, these Laplacians are well-defined only for strongly connected directed graphs.}
It is important to note that these operators are no longer linear, and therefore, cannot be expressed using matrices.
Furthermore, recently,~\cite{Yoshida:2017zz} showed that these operators can be systematically constructed using the cut function $F\colon 2^V \to \bbR_+^E$ associated with an undirected graph, directed graph, or hypergraph, where $F_e\colon S \mapsto F(S)(e)$ is equal to one if the edge $e \in E$ is cut, and zero otherwise.
A key property used in the analysis of the constructed operators is the \emph{submodularity} of the cut function, which is given by
\[
  F_e(S) + F_e(T) \geq F_e(S \cup T) + F_e(S \cap T)
\]
for every $S,T \subseteq V$.
Indeed, this construction can be applied to any \emph{submodular transformation} $F\colon2^V \to \bbR_+^E$, that is, each function $F_e\colon S \mapsto F(S)(e)$ is non-negative submodular, and corresponding Cheeger's inequality was obtained when $F(\emptyset)=F(V)=\bmzero$~\citep{Yoshida:2017zz}.
In what follows, we always assume that a submodular transformation is \emph{normalized}, that is, $F(\emptyset)=\bmzero$, and $F(V)=\bmzero$ to avoid technical triviality.
% We do not impose the constraint $F(V)=\bmzero$ in this work. \ynote{is this true?}

\paragraph{Submodular Laplacian Systems.}

Because solving Laplacian systems based on undirected graphs have numerous applications, it is natural to consider solving Laplacian systems for general submodular transformations, which are referred to as \emph{submodular Laplacian systems}.
It should be noted that the Laplacian operator $L_F\colon\bbR^V \to 2^{\bbR^V}$ associated with a submodular transformation $F\colon2^V \to \bbR_+^E$ is multi-valued, and therefore, the question here is computing $\bmx \in \bbR^V$ such that $L_F(\bmx) \ni \bmb$ for a given vector $\bmb \in \bbR^V$.
As the first contribution of our study, we show the tractability of submodular Laplacian systems:
\begin{theorem}\label{the:intro-system}
  Let $F\colon2^V \to \bbR_+^E$ be a submodular transformation and $\bmb \in \bbR^V$ be a vector.
  Then, we can compute $\bmx \in \bbR^V$
  \[
    L_F(\bmx) \ni \bmb
  \]
  in polynomial time, if it exists.
%  Moreover, if each function $F_e\colon2^V \to \bbR$ is of constant arity, the time complexity can be improved to $O(|V|^3)$.
\end{theorem}
Because the algorithm used in Theorem~\ref{the:intro-system} is based on the ellipsoid method, the time complexity is large albeit polynomial.
Later, in Section~\ref{subsec:system-constant-arity}, we will  discuss more efficient algorithms when $F$ is given by a directed graph or hypergraph.

In some special cases of the abovementioned problem, a direct interpretation is possible, which is discussed in the following lines.
Suppose $F\colon2^V \to \bbR_+^E$ is constructed from a directed graph $G=(V,E)$.
Then, we regard each arc $uv \in E$ as an ideal diode of $1\Omega$, that is, if the electric potential of $u$ is higher than or equal to that of $v$, then each arc represents a resistance of $1\Omega$, and otherwise, no current flows through it.
Then, the solution $\bmx$ for this directed graph provides the electric potential at the vertices when a current of $\bmb(v)$A flows through each $v \in V$.
For hypergraphs, a hyperedge acts as a circuit element in which current flows from a vertex of the highest potential to that of the lowest.

Solving Laplacian systems for undirected graphs has been intensively studied.
The first nearly-linear-time algorithm was achieved by~\cite{Spielman:2014gq}, and the current fastest algorithm runs in $\widetilde{O}(|E|\log^{1/2}|V|)$ time~\citep{Cohen:2014do}, where $\widetilde{O}(\cdot)$ hides a polylogarithmic factor.
In contrast, to the best of our knowledge, there is no prior study on solving Laplacian systems for directed graphs and hypergraphs.

\paragraph{Submodular Laplacian Regression.}

When the given submodular Laplacian system $L_F(\bmx) \ni\bmb$ does not admit a solution, we may want to find $\bmb' \in \bbR^V$ close to $\bmb$ such that $L_F(\bmx) \ni \bmb'$ has a solution;
this serves as motivation for the following \emph{submodular Laplacian regression problem}:
\begin{align*}
  \min_{\bmp \in \bbR^V} \|\bmp\|_2^2 \quad
  \text{subject to } L_F(\bmx) \ni \bmb+\bmp \text{ for some } \bmx \in \bbR^V.
\end{align*}
It should be noted that once we have solved this problem, we can obtain $\bmx \in \bbR^V$ with $L_F(\bmx) \ni \bmb+\bmp$ in polynomial time by applying Theorem~\ref{the:intro-system}.
The second contribution of this study is the following:
\begin{theorem}\label{the:intro-regression}
%  Let $F:2^V \to \bbR^E$ be a submodular transformation and $\bmb \in \bbR^V$ be a vector.
  We can solve  the submodular Laplacian regression problem in polynomial time.
\end{theorem}

%when constructing sparsifiers~\cite{Spielman:2011ir} and used in network analysis~\cite{Hayashi:2016ub,Mavroforakis:2015to,Newman:2005em}.

\subsection{Applications}
In this section, we discuss potential applications of our results in machine learning and network analysis.

\paragraph{Semi-supervised Learning}

Semi-supervised learning is a framework for predicting unknown labels of unlabeled data based on both labeled and unlabeled instances. 
In the case of supervised learning, we utilize only labeled instances to predict the labels of unlabeled instances. 
However, in many realistic scenarios, only a few of labeled instances are available compared with a large number of unlabeled instances; 
in such cases, the unlabeled instances also need to be effectively utilized to predict the unknown labels.
This type of learning is referred to as semi-supervised learning.

Formally, the problem of semi-supervised learning can be stated follows: The given dataset $V$ is partitioned into a set $T$ of labeled instances and set $U$ of unlabeled instances. 
Each labeled instance $v \in T$ has its label $\tilde{\bmx}(v) \in \{-1, +1\}$. 
The aim of this problem is to predict the labels $\bmx(v)$ of unlabeled instances $v \in U$.

A standard approach to semi-supervised learning is using Laplacians associated with undirected graphs.
In this method, first, the training dataset is transformed into a similarity graph $G =(V, E, w)$, where $w \colon E \to \bbR$ is a weight function; then, the labels of unlabeled instances are predicted using the Laplacian matrix $L_G \in \bbR^{V \times V}$ for the constructed similarity graph. 
In particular, a similarity graph represents pairwise representations among labeled and unlabeled instances.

Another well known approach proposed by~\cite{Zhu:2003tb} involves solving the Laplacian system $L_G \bmx = \bmb$ under the constraints that $\bmb(v) = 0$ for all $v \in U$ and $\bmx(v) = \tilde{\bmx}(v)$ for all $v \in T$. 
Since its introduction, many variants of this approach have been proposed based on Laplacians for undirected graphs, which have applied to various applications. %\ynote{I think it's better to explicitly write the extension of Theorem~\ref{the:intro-system} with these constraints using the theorem environment and move Appendix A to the main text.}
%\ynote{can we (directly) solve the latter using our results?}\knote{I found the difference between Laplacian regularization and Laplacian system is not so small. But I think it is also solved with the ellipsoid method if the subdifferential of the loss function is computable.}

Using submodular Laplacian systems, we can extend previous semi-supervised learning algorithms to more general similarity expressions, such as directed graphs and hypergraphs.
An arc represents that the tail vertex is closer to $+1$ than the head vertex.
A hyperedge represents that all incident vertices have similar labels.
Furthermore, we can express more complicated relationships with general submodular functions.

To develop a general framework for semi-supervised learning, we extend Theorem~\ref{the:intro-system} by imposing constraints $\bmx(v) = \tilde{\bmx}(v)\;(v \in T)$ and $\bmb(v) = \tilde{\bmb}(v)\;(v \in U)$, where $\tilde{\bmx} \in \bbR^T$and  $\tilde{\bmb} \in \bbR^U$ are provided as a part of the input.
It should be noted that the semi-supervised setting is a special case of this problem for which $\tilde{\bmb}$ is the zero vector.
Then, we show the following:
%This setting is more general than semi-supervised setting in the
%, that is, the constraints on $\bmb$ is that $\bmb(v) = \tilde{\bmb}(v)$ for any $v \in U$ where $\tilde{\bmb} \in \bbR^U$ is a given vector.
%Our contribution is showing this problem can be solved in polynomial time.
\begin{theorem}\label{the:intro-semi-supervised}
	Let $F \colon 2^V \to \bbR^E_+$ be a submodular transformation and suppose $V$ is partitioned into disjoint subsets $T \subseteq V$ and $U \subseteq V$.
  Let $\tilde{\bmx} \in \bbR^T$ and $\tilde{\bmb} \in \bbR^U$ be vectors.
  Then, we can compute $\bmx \in \bbR^V$ that satisfies $L_F(\bmx) \ni \bmb$, $\bmx(v) = \tilde{\bmx}(v)$ for all $v \in T$ and $\bmb(v) = \tilde{\bmb}(v)$ for all $v \in U$ in polynomial time if it exists.
  % Moreover, if each function $F_e \colon 2^V \to \bbR_+$ is of constant arity, then the time complexity can be improved to $O(|V|^3)$.
\end{theorem}

\paragraph{Network Analysis}

A typical task in network analysis is measuring the importance, or \emph{centrality}, of vertices and edges.
There are many centrality notions depending on applications, and for undirected graphs, some notions are defined using Laplacians matrices.

Before describing these centrality notions, we require some definitions.
Let $G=(V,E)$ be an undirected graph, and for each $v \in V$, let $\bme_v \in \bbR^V$ be the vector with $\bme_v(v) = 1$ and $\bme_v(w) = 0$ for $w \in V \setminus \set{v}$.
Then, for vertices $u,v \in V$, the quantity $R_G(u,v):=(\bme_u-\bme_v)^\top L_G^+(\bme_u-\bme_v) \in \bbR_+$ is called the \emph{effective resistance} from $u$ to $v$, where $L_G^+ \in \bbR^{V \times V}$ is the pseudo-inverse of the Laplacian $L_G\in \bbR^{V \times V}$.
The effective resistance can be considered as the resistance of the circuit associated with $G$ when a current is passed from $u$ to $v$.
%It is also known that $2|E| R_{uv}$ is the commuting time of a random walk between $u$ and $v$.

In~\citep{Hayashi:2016ub,Mavroforakis:2015to}, the effective resistance $R_G(u,v)$ of an edge $uv \in E$ is directly used as the centrality of an edge, and is referred to as the \emph{spanning tree centrality}, because it is known to be equal to the probability that the edge is used when sampling a spanning tree uniformly at random.

\cite{Brandes:2005bk} introduced the \emph{current flow closeness centrality} of a vertex, which is defined as
\[
  \tau_{\mathrm{C}}(v) = \frac{n}{\sum_{u \in V: u\neq v}R_G(u,v)}.
\]
This notion implicitly assumes that the effective resistance satisfies the \emph{triangle inequality}, that is, $R_G(u,v)+R_G(v,w) \geq R_G(u,w)$ for every $u,v,w\in V$ and hence can be seen as a (quasi-)metric.
Then, we regard a vertex as important when it is close to every other vertex; in other words, it is in the center of the network.

\cite{Newman:2005em} introduced the notion of the \emph{current flow betweenness centrality} of a vertex, which is defined as
\[
  \tau_{\mathrm{B}}(v) = \frac{1}{(|V|-1)(|V|-2)}\sum_{s,t \in V: v \neq s \neq t \neq v}\tau_{st}(v),
\]
where $\tau_{st}(v)=\sum_{w:vw\in E}|\bmx(v)-\bmx(w)|$ for $\bmx = L_G^+(\bme_s-\bme_t)$.
In words, $\tau_{st}(v)$ is the total current passing through the edges incident to $v$ when a current of $1A$ flows from $s$ to $t$.
Intuitively, we regard a vertex as important when a large fraction of the current passes through it.

It is natural to inquire whether these notions can be extended to directed graphs, hypergraphs, or general submodular transformations.
Using Theorem~\ref{the:intro-system}, we can define the effective resistance $R_F(u,v)$ for $u,v \in V$ in terms of the Laplacian $L_F\colon\bbR^V \to 2^{\bbR^V}$ associated with a submodular transformation $F\colon2^V \to \bbR_+^E$.
That is, we define
\[
  R_F(u,v):=(\bme_u-\bme_v)^\top L_F^+(\bme_u-\bme_v),
\]
where $L_F^+(\bme_u-\bme_v) \in \bbR^V$ is the vector $\bmx \in \bbR^V$ obtained by applying Theorem~\ref{the:intro-system} with $\bmb=\bme_u-\bme_v$.
When there is no solution, we define $R_F(u,v) = \infty$.
The formal definition is deferred to Section~\ref{sec:effective-resistance}.

For directed graphs and hypergraphs, as in the case of undirected graphs, the effective resistance $R_F(u,v)$ can be seen as the resistance of the circuit associated with the graph when a current is passed from $u$ to $v$.
%As opposed to the undirected graph case, there could be several flow configurations that satisfies Kirchhoff and Ohm's laws, but we choose the one with the minimum energy.
Note that $R_F(u,v)$ may not be equal to $R_F(v,u)$ for directed graphs.

Further, we can generalize centrality notions mentioned above to general submodular transformations.
As mentioned earlier, the implicit assumption when defining current flow closeness centrality is of the triangle inequality of effective resistance.
Here, we show that it also holds for general submodular transformations:
\begin{theorem}\label{the:intro-triangle-inequality}
  Let $F\colon2^V \to \bbR_+^E$ be a submodular transformation.
  Then, effective resistance $R_F(u,v)$ satisfies triangle inequality.
\end{theorem}
% Hence, we can naturally define current flow closeness centrality for general submodular transformations.

\subsection{Organization}
The notions used in this paper are reviewed in Section~\ref{sec:pre}.
Then, we prove Theorems~\ref{the:intro-system} and~\ref{the:intro-regression} in Sections~\ref{sec:system} and~\ref{sec:regression}, respectively.
Due to space limitations, the proofs of Theorems~\ref{the:intro-semi-supervised} and~\ref{the:intro-triangle-inequality} are deferred to Sections~\ref{sec:semi-supervised} and~\ref{sec:effective-resistance}, respectively.

%!TEX root=./main.tex
\section{Preliminaries}\label{sec:pre}

For a positive integer $n \in \bbN$, the set $\set{1,\ldots,n}$ is denoted by $[n]$.
For a vector $\bmx \in \bbR^V$ and a set $S \subseteq V$, $\bmx(S)$ denotes the sum $\sum_{v \in S}\bmx(v)$.
%For an ordering $\pi=(\pi_1,\ldots,\pi_{|V|})$ of $V$, we define $\bbR^V_\pi$ as the set $\set{\bmx \in \bbR^V \mid \bmx(\pi_1) \geq \bmx(\pi_2) \geq \cdots \geq \bmx(\pi_{|V|})}$.

\subsection{Distributive lattice and the Birkoff representation theorem}
A subset $S$ in a poset $(X, \preceq)$ is called a \emph{lower ideal} if $S$ is down-closed with respect to the partial order $\preceq$.
% Given a directed graph $D$, a subset $X \subseteq V$ is called a \emph{lower ideal} if any vertex reachable from a vertex in $X$ is also contained in $X$.
The \emph{Birkoff representation theorem}~\citep{Birkhoff1937,Fujishige2005} states that for any distributive lattice $\caD \subseteq 2^V$ with $\emptyset, V \in \caD$, there uniquely exists a poset $P$ on a partition $\Pi(V)$ of $V$ and a partial order $\preceq$ such that
$S \in \caD$ if and only if $S$ is a lower ideal in $P$.
The poset $P$ can be constructed as follows.
Let us define an equivalence relation $\sim$ on $V$ as $i \sim j$ if any $S \in \caD$ contains either both, $i$ and $j$, or none.
Let $\Pi(V)$ be the set of equivalence classes induced by $\sim$.
Then, a partial order $\preceq$ on $\Pi(V)$ is defined as $[i] \preceq [j]$, where $[i]$ and $[j]$ are equivalence classes containing $i$ and $j$, respectively, if any set $S \in \caD$ containing $j$ also contains $i$.
It can be confirmed that $\preceq$ is a well-defined partial order.
Then, $P = (\Pi(V), \preceq)$ is the desired poset.

Using the Hasse diagram of $P$, one can maintain the poset $P$ in a digraph with $n$ vertices and $m$ arcs.
We refer to such a digraph as the \emph{Birkoff representation of $\caD$}.
Since $n = O(|V|)$ and $m = O(n^2)$, the Birkoff representation is a compact representation of a distributive lattice, even if $\caD$ contains exponentially many subsets.
Furthermore, if $\caD$ is the set of minimizers of a submodular function, one can construct the Birkoff representation in strongly polynomial time using submodular function minimization algorithms (see \citet{Fujishige2005}).

%\ynote{we may want to move this to Section~\ref{sec:regression} as it is only used there.}

% poset on a partition ($\Pi(V), \preceq)$ \ynote{Explain what $E$ is. By the way, we may want to write $V$ here because the domains of submodular functions we use is always $2^V$.} such that any $X \in \caD$ uniquely corresponds to a lower ideal of $(\Pi(E), \preceq)$.
% Indeed, the partial order is defined as follows: .

\subsection{Submodular functions}
Let $\caD \subseteq 2^V$ be a distributive lattice.
A function $F\colon \caD \to \bbR$ is called \emph{submodular} if
\[
  F(S) + F(T) \geq F(S \cup T) + F(S \cap T)
\]
for every $S,T \in \caD$.
For a submodular function $F\colon \caD \to \bbR$, the \emph{submodular polyhedron} $P_\caD(f)$ and \emph{base polyhedron} $B_\caD(F)$ associated with $F$ are defined as
\begin{align*}
    P_\caD(F) & = \{ \bmx \in \bbR^V : \bmx(S) \leq F(S)\;(S \in \caD) \}  \\
    B_\caD(F) & = \{ \bmx \in P_\caD(F) : \bmx(E) = F(V) \},
\end{align*}
respectively.
We omit the subscripts $\caD$ when it is clear from the context.

% \begin{lemma}[{\cite[Theorem~3.2]{Fujishige2005}}]
    % A base polytope $B_\caD(F)$ is bounded if and only if $\caD = 2^E$.
% \end{lemma}

For a submodular function $F\colon \caD \to \bbR$, its \emph{\Lovasz extension} $f\colon \bbR^V \to \bbR$ is defined as
\[
  f(\bmx) = \max_{\bmw \in B_\caD(F)}\langle \bmx,\bmw\rangle,
\]
where $\langle \cdot ,\cdot\rangle$ is the inner product.
We note that $f$ is an extension considering $f(\bmone_S) = F(S)$ for every $S \in \caD$, where $\bmone_S \in \bbR^V$ is a vector with $\bmone_S(i) = 1$ if $i \in S$ and $\bmone_S(i) = 0$ otherwise.
Moreover, $f$ is convex and positively homogeneous, that is, $f(\alpha \bmx) = \alpha f(\bmx)$ for every $\alpha \geq 0$.
It is known that the subdifferential of $f$ at $\bmx$ is $\argmax_{\bmw \in B_\caD(F)}\langle \bmx,\bmw\rangle$, and this is denoted by $\partial f(\bmx)$.

\subsection{Submodular Laplacians}

A transformation $F\colon 2^V \to \bbR_+^E$ is called \emph{submodular} if each function $F_e\colon S \mapsto F(S)(e)$ is submodular.
We now formally introduce the Laplacian operator associated with a submodular transformation:
\begin{definition}[Submodular Laplacian operator~\citep{Yoshida:2017zz}]\label{def:submodular-Laplacian}
  Let $F\colon 2^V \to \bbR_+^E$ be a submodular transformation.
%  For each $e \in E$, let $F_e:2^{[n]} \to \bbR$ be the submodular function corresponding to the $j$-th component of $f$.
  Then, the Laplacian $L_F\colon \bbR^V \to 2^{\bbR^V}$ of $F$ is defined as
  \[
    L_F(\bmx) = \Bigset{\sum_{e \in E} \bmw_e \langle \bmw_e, \bmx\rangle : \bmw_e \in \partial f_e(\bmx)\; (e \in E)} = \Bigset{W W^\top \bmx : W \in \prod_{e \in E}\partial f_e(\bmx)},
  \]
  where $f_e\colon \bbR^V \to \bbR$ is the \Lovasz extension of $F_e$ for each $e \in E$.
\end{definition}
Here, we have $\langle \bmx,\bmy\rangle = \sum_{e \in E}f_e(\bmx)^2$ for any $\bmy \in L_F(\bmx)$, and hence we can write $\bmx^\top L_F(\bmx)$ to denote this quantity.

As a descriptive example, consider an undirected graph $G=(V,E)$ and the associated submodular transformation $F:2^V \to \bbR_+^E$.
Then, for every edge $e = uv \in E$ we have
\[
\partial f_e(\bmx) = \begin{cases}
  \set{\bme_u- \bme_v} & \text{if }\bmx(u) > \bmx(v), \\
  \set{\bme_v- \bme_u} & \text{if }\bmx(u) < \bmx(v), \\
  \bigset{\alpha \bme_u- \alpha\bme_v : \alpha \in [-1,1]} & \text{if } \bmx(u)=\bmx(v).
\end{cases}
\]
We can confirm that $L_F(\bmx)(v) = \sum_{w:vw\in E}\bigl(\bmx(v)-\bmx(w)\bigr)= L_G(\bmx)(v)$ for every $v \in V$, where $L_G$ is the standard Laplacian matrix associated with $G$.
Refer to~\citep{Yoshida:2017zz} for further examples.

\subsection{Convex optimization}
We use the following basic results from convex optimization.

\begin{theorem}[{\citet[Corollary~28.3.1]{Rockafellar1996}}]\label{thm:saddle}
    Let $f, g_1, \dots, g_m : \bbR^n \to \bbR$ be proper convex functions.
    Consider the optimization problem
    \begin{align*}
        \min_{\bmx \in \bbR^n} f(\bmx) \quad \text{subject to} \quad g_i(\bmx) \leq 0. \quad (i = 1, \dots, m)
    \end{align*}
    Assume that the optimal value is finite and the Slater condition is satisfied.
    Then, a feasible solution $\bmx$ is optimal if and only if there exists a Lagrange multiplier $\bmphi \in \bbR_+^m$ such that $(\bmx, \bmphi)$ is a saddle point of the Lagrangian
    \begin{align*}
        P(\bmx, \bmphi) := f(\bmx) + \sum_{i \in [m]} \bmphi(i)g_i(\bmx).
    \end{align*}
    Equivalently, $\bmx$ is optimal if and only if there exists $\bmphi \in \bbR^m_+$ satisfying the Karush-Kuhn-Tucker (KKT) condition together with $\bmx$.
\end{theorem}

%!TEX root=./main.tex

\section{Solving Submodular Laplacian Systems}\label{sec:system}

In this section, we prove Theorem~\ref{the:intro-system}.
Throughout this section, we fix a submodular transformation $F\colon \bbR^V \to \bbR_+^E$ and a vector $\bmb \in \bbR^V$.
%Recall that the Laplacian $L_F$ is a multi-valued function.
%Hence, the goal is to find $\bmx \in \bbR^V$ such that $L_F(\bmx) \ni \bmb$.
We describe a general method based on the ellipsoid method in Section~\ref{subsec:system-general}.
Then, in Section~\ref{subsec:system-constant-arity}, we discuss a more efficient algorithm when $F\colon 2^V \to \bbR_+^E$ is given by a directed graph or hypergraph.

\subsection{General case}\label{subsec:system-general}
Let $f_e\colon \bbR^V \to \bbR\;(e\in E)$ be the \Lovasz extension of $F_e$.
We consider the following optimization problem:
\begin{align}\label{eq:general-primal}
  \min_{\bmx \in \bbR^V,\bmeta \in \bbR^E} \frac{1}{2}\| \bmeta \|_2^2  - \langle \bmb, \bmx \rangle \quad \text{subject to } f_e(\bmx) \leq \bmeta(e) \quad (e \in E).
\end{align}
This problem is equivalent to minimizing $\frac{1}{2}\sum_{e \in E}f_e(\bmx)^2  - \langle \bmb, \bmx \rangle = \frac{1}{2}\bmx^\top L_F(\bmx) - \langle \bmb,\bmx\rangle$, which is continuously differentiable and convex.
Then, we can directly show that the minimizer $\bmx^*$ of the latter problem satisfies $L_F(\bmx) \ni \bmb$ from the first order condition.
By considering~\eqref{eq:general-primal}, however, we can exploit the combinatorial structure of the problem to obtain a more efficient algorithm.

Introducing the Lagrange multiplier $\bmphi \geq \bmzero$, we can obtain the corresponding Lagrangian as
\begin{align}
  P(\bmx,\bmeta,\bmphi)
  &= \frac{1}{2}\|\bmeta\|_2^2 -\langle \bmb,\bmx\rangle + \sum_{e \in E} \bmphi(e) (f_e(\bmx) - \bmeta(e)) \label{eq:general-lagrangian} \\
  &= \frac{1}{2}\|\bmeta\|_2^2  -\langle \bmphi, \bmeta \rangle - \left[ \langle \bmb, \bmx \rangle - \sum_{e \in E} \bmphi(e) f_e(\bmx) \right]. \nonumber
\end{align}

Thus,
\begin{align*}
  \min_{\bmx, \bmeta}P(\bmx,\bmeta,\bmphi)
  &= -\frac{1}{2}\|\bmphi\|_2^2 - \max_{\bmx}\left[ \langle \bmb, \bmx \rangle - \sum_{e \in E} \bmphi(e) f_e(\bmx) \right]
  = -\frac{1}{2}\|\bmphi\|_2^2 - \left(\sum_{e \in E} \bmphi(e) f_e \right)^*(\bmb),
\end{align*}
where $(\cdot)^*$ is the Fenchel conjugate.
Since \Lovasz extensions are positively homogeneous, we can confirm that $\left(\sum_{e \in E} \bmphi(e) f_e \right)^*(\bmb)$ is equal to either infinity or zero.
Thus, we obtain the dual problem:
\begin{align}\label{eq:general-dual}
  \min_{\bmphi \geq \bmzero} \frac{1}{2}\| \bmphi \|_2^2 \quad \text{subject to } \left(\sum_{e \in E} \bmphi(e) f_e \right)^*(\bmb) \leq 0.
\end{align}
%\ynote{where do we get $\bmphi \geq \bmzero$?}\tnote{Because $\phi$ corresponds to inequality constraints.}

The constraint can be checked with submodular function minimization.
To see this, we observe the following:
\begin{lemma}\label{lem:sup-and-max}
  We have
  \begin{align*}
    \left(\sum_{e \in E} \bmphi(e) f_e \right)^*(\bmb) = \sup_{\bmx\in\bbR^V} \left[ \inprod{\bmb, \bmx} - \sum_{e\in E}\bmphi(e) f_e(\bmx) \right] \leq 0
  \end{align*}
  if and only if
  \begin{align*}
    \max_{\bmx\in [0,1]^V} \left[ \inprod{\bmb, \bmx} - \sum_{e\in E}\bmphi(e) f_e(\bmx) \right]
      = \max_{X \subseteq V} \left[ \bmb(X) - \sum_{e\in E}\bmphi(e) F_e(X) \right] \leq 0.
  \end{align*}
\end{lemma}
\begin{proof}
  $(\Rightarrow)$ Trivial.

  $(\Leftarrow)$ Suppose there is $\bmx \in \bbR^V$ such that $\inprod{\bmb, \bmx} - \sum_{e\in E}\bmphi(e) f_e(\bmx) > 0$.
  By the assumption that $F(V) = \bmzero$, we have $f(\bmx) = f(\bmx + \alpha \bmone)$ for any $\alpha \in \bbR$, where $\bmone \in \bbR^V$ is the all-one vector.
  This holds because $f_e(\bmx+\alpha \bmone) = \max_{\bmw \in B(F_e)}\langle \bmw,\bmx+\alpha \bmone\rangle =  \max_{\bmw \in B(F_e)}\langle \bmw,\bmx\rangle = f_e(\bmx)$, where we used the fact that $\langle \bmw,\bmone\rangle = 0$ holds for any $\bmw \in B(F_e)$ as $F_e(V)=0$.

  Now, by adding $\alpha \bmone$ to $\bmx$ for a large $\alpha \in \bbR$, we can assume $\bmx(v) \geq 0$ for every $v \in V$.
  Moreover, $\bmx \neq \bmzero$ because $f(\bmzero) = \bmzero$.
  Since the \Lovasz extension $f$ is positively homogeneous, for $\bmx' = \bmx/\|\bmx\|_\infty \in [0,1]^V$, we have $f(\bmx') = f(\bmx/\|\bmx\|_\infty) = f(\bmx)/\|\bmx\|_\infty > 0$.
\end{proof}

By Lemma~\ref{lem:sup-and-max}, we obtain the following dual problem:
\begin{align}\label{eq:general-energy-minimization}
  \min_{\bmphi \geq \bmzero} \frac{1}{2}\| \bmphi \|_2^2  \quad \text{subject to } \sum_{e \in E} \bmphi(e)F_e(X) - \bmb(X) \geq 0 \; (X \subseteq V).
\end{align}
A separation oracle for the constraint can be implemented by submodular function minimization.
Therefore, we can use the ellipsoid method to solve~\eqref{eq:general-energy-minimization}.

\begin{theorem}\label{the:general}
  The following hold:
  \begin{itemize}
  \itemsep=0pt
  \item[(1)] \eqref{eq:general-primal} is feasible if and only if~\eqref{eq:general-energy-minimization} is bounded.
  \item[(2)] A strong duality holds between~\eqref{eq:general-primal} and~\eqref{eq:general-energy-minimization}, that is, the optimal values of~\eqref{eq:general-primal} and~\eqref{eq:general-energy-minimization} coincide.
  \item[(3)] We can solve~\eqref{eq:general-primal} in polynomial time.
  \item[(4)] The optimal solution $(\bmx^*,\bmeta^*)$ for~\eqref{eq:general-primal} satisfies $L_F(\bmx^*) \ni \bmb$.
  \end{itemize}
\end{theorem}
\begin{proof}
  (1) Standard.

  (2) It is clear that~\eqref{eq:general-primal} satisfies Slater's condition and hence the claim holds.

  (3) Since we can compute the subgradient of $f_e$ at a given point $\bmx$, we have a separation oracle for~\eqref{eq:general-primal}. Thus the ellipsoid method solves~\eqref{eq:general-primal}.

  (4) Let $(\bmx^*, \bmeta^*)$ be an optimal solution of~\eqref{eq:general-primal}.
  Based on Theorem~\ref{thm:saddle}, there exists $\bmphi^* \geq \bmzero$ such that $(\bmx^*, \bmeta^*, \bmphi^*)$ is a saddle point of the Lagrangian $P$.
  Since
  $
      \frac{\partial P}{\partial \bmeta} = \bmeta - \bmphi,
  $
  we have $\bmphi^* = \bmeta^*$ using the saddle condition.
  By complementary slackness, if $\bmphi^*(e) > 0$, we have $f_e(\bmx^*) = \bmeta^*(e)$ for any $e \in E$.
  If $\bmphi^*(e) = 0$, we have $f_e(\bmx^*) \leq 0$, and since $f_e(\bmx^*) \geq 0$, we obtain $f_e(\bmx^*) = \bmeta^*(e)$ for $e \in E$.
  Thus, we conclude that $\bmphi^*(e) = f_e(\bmx^*)$ for $e \in E$.
  Finally, by the saddle condition for $\bmx^*$, we must have
  \[
      \sum_{e \in E} \bmphi^*(e) \partial f_e(\bmx^*) = \sum_{e \in E} f_e(\bmx^*)\partial f_e(\bmx^*) \ni \bmb,
  \]
  which implies $\bmb \in L_F(\bmx^*)$.
\end{proof}

\subsection{Flow-like formulation for directed graphs and hypergraphs}\label{subsec:system-constant-arity}
If each submodular function $F_e\colon 2^V \to \bbR_+$ is of constant arity, we can use a different formulation.
First, we enumerate all extreme points of the base polytope $B(F_e)$ for each $F_e$.
Let $\caV_e$ be the set of extreme points of $B(F_e)$ ($e \in E$).
The constraint $f_e(\bmx) \leq \bmeta(e)$ in~\eqref{eq:general-primal} is equivalent to $\inprod{\bmw, \bmx} \leq \bmeta(e)$ for all extreme points $\bmw \in \caV_e$.
Therefore,~\eqref{eq:general-primal} is equivalent to
\begin{align}\label{eq:flowlike-primal}
    \min_{\bmx \in \bbR^V,\bmeta \in \bbR^E} \frac{1}{2}\norm{\bmeta}_2^2 - \inprod{\bmb, \bmx} \quad \text{subject to} \quad \inprod{\bmw, \bmx} \leq \bmeta(e) \quad (\bmw \in \caV_e, e \in E)
\end{align}

For each extreme point $\bmw \in \caV_e$, we introduce a ``flow'' variable $\bmphi(e, \bmw)$.
By a calculation similar to that in Section~\ref{subsec:system-general}, we can obtain the following dual problem:
\begin{align}\label{eq:flowlike-dual}
    \min_{\bmphi \geq \bmzero} \frac{1}{2} \sum_{e\in E} \left( \sum_{\bmw \in \caV_e} \bmphi(e, \bmw) \right)^2 \quad \text{subject to} \quad \sum_{e \in E}\sum_{\bmw \in \caV_e} \bmphi(e, \bmw)\bmw = \bmb.
\end{align}
Then, the constraint $\sum_{e \in E}\sum_{\bmw \in \caV_e} \bmphi(e, \bmw)\bmw = \bmb$ can be interpreted as a ``flow boundary constraint'', as illustrated in the following examples.

\begin{example}[Cut functions of directed graphs]
    Let $G = (V, E)$ be a directed graph and $F_e:2^V \to \bbR_+\;(e \in E)$ be the cut function associated with $e$.
    The extreme points of $B(F_e)$ for $e=uv$ are $\bmzero$ and $\bme_u - \bme_v$.
    Since the value of $\bmphi(e, \bmzero)$ does not interact with the constraint, we can assume that $\bmphi(e, \bmzero) = 0$.
    Then, the constraint is the ordinary flow boundary constraint: $\sum_{uv \in E} \bmphi(uv) - \sum_{vu \in E} \bmphi(vu) = \bmb(u)$ ($u \in V$), where we denote $\bmphi(uv, \bme_v - \bme_u)$ by $\bmphi(uv)$.
    Now~\eqref{eq:flowlike-dual} is equivalent to the quadratic cost flow problem, which can be solved in $O(|E|^4 \log |E|)$ time~\citep{Vegh:2012jm}.
\end{example}

\begin{example}[Cut functions of hypergraphs]
    Let $G = (V, E)$ be a hypergraph and $F_e\colon 2^V \to \bbR_+\;(e \in E)$ be the cut function associated with $e$.
%    Without loss of generality, one can assume $e = \{1, 2, \dots, d\}$.
    The extreme points of $B(F_e)$ are in the form of $\bme_u - \bme_v$ ($u, v \in e$, $u \neq v$).
    The value of $\bmphi(e, \bme_u - \bme_v)$ can be interpreted as a ``flow'' from $v$ to $u$  through a hyperedge $e$. %\ynote{from $u$ to $v$?}\tnote{Since $\bme_u$}
    Indeed, we can construct the equivalent (ordinary) flow network $G'$ as follows.
    The vertex set of $G'$ is $V$, and for each distinct $u, v \in e$, an arc $uv$ in $G'$ is drawn.
    Then, any flow $\bmphi'$ in $G'$ with the boundary $\bmb$, which can be computed via minimizing a quadratic function under a flow constraint, corresponds to the original variable $\bmphi$ in~\eqref{eq:flowlike-dual}.
    % Therefore,~\eqref{eq:flowlike-dual} is equilvalent to the quadratic flow problem in $G'$.
\end{example}

\section{Submodular Laplacian Regression}\label{sec:regression}

In this section, we prove Theorem~\ref{the:intro-regression}.

First, we explain when a submodular Laplacian system, or equivalently~\eqref{eq:general-energy-minimization}, is feasible.
Let $F\colon 2^V \to \bbR_+^E$ be a submodular transformation and $\bmb \in \bbR^V$ be a vector.
Then, we can observe that~\eqref{eq:general-energy-minimization} is feasible if and only if
there exists no $S \subseteq V$ such that $F_e(S) = 0$ for every $e \in E$ and $\bmb(S) > 0$, or equivalently,
$F_e(S) = 0$ for every $e \in E$ implies $\bmb(S) \leq 0$.
We define $\ker(F) := \set{ S \subseteq V: F_e(S) = 0 \; (e\in E)}$, which is the set of $S \subseteq V$ that minimizes all $F_e$ ($e \in E$).
Then, the regression problem reduces to the following optimization problem.
\begin{align}\label{eq:minnorm-polyhedron}
  \min_{\bmp \in \bbR^V} \|\bmp\|_2^2 \quad \text{subject to } (\bmb+\bmp)(S) \leq 0 \; (S \in \ker(F))
\end{align}
Since the minimizers of each $F_e$ form a distributive lattice, $\ker(F)$ is also a distributive lattice.
Hence,~\eqref{eq:minnorm-polyhedron} can be considered as the minimum norm point problem in $P_\caD(-\bmb)$, where $\caD = \ker F$.

To handle~\eqref{eq:minnorm-polyhedron} efficiently, we use the Birkoff representation of $\ker(F)$ because it has a polynomial size.
To this end, note that $\ker(F)$ is the lattice of minimizers of $\sum_{e \in E} F_e$, since each $F_e$ is nonnegative.
Then, the Birkoff representation can be efficiently constructed from the minimum norm point of $B(\sum_e F_e)$.
Refer to~\cite[Section~7.1~(a)]{Fujishige2005} for further details.

The minimum norm point problem~\eqref{eq:minnorm-polyhedron} is slightly different from the one used for submodular function minimization considering that
(i) the target polytope is a submodular polyhedron $P_\caD(-\bmb)$ rather than a base polyhedron $B_\caD(-\bmb)$, and
(ii) the lattice $\caD$ is not a Boolean lattice $2^V$ but a distributive lattice on $V$.
In the following subsection, we present two algorithms for solving this problem.
The first one is the standard Frank-Wolfe iterative algorithm (Section~\ref{subsec:frank-wolfe}), while the other is a  combinatorial algorithm (Section~\ref{subsec:combinatorial}).

\subsection{Frank-Wolfe algorithm}\label{subsec:frank-wolfe}
Given the Birkoff representation of $\caD = \ker F$ as a directed graph with $n$ vertices and $m$ arcs, we can optimize linear functions over the polyhedron $P_\caD(-\bmb)$ in $O(m + n\log {n})$ time using the greedy algorithm.
%\ynote{The time complexity should depend on the size of $\caD$, which can be $|V|^2$.}
This fact suggests to use the Frank-Wolfe algorithm~\citep{Jaggi2013} to solve~\eqref{eq:minnorm-polyhedron}.
To this end, we restrict ourselves to a compact region $\caC := P_\caD(-\bmb) \cap \set{\bmx \in \bbR^V: \bmx \geq -\bmb }$.
Based on the analysis of the Frank-Wolfe algorithm in~\cite{Jaggi2013}, we need to bound the squared Euclidean diameter of $\caC$.
Trivially, $|V| \cdot \norm{-\bmb}_2^2$ is an upper bound.
We obtain the following convergence rate:
\begin{theorem}\label{the:regression-by-Frank-Wolfe}
    Let $\bmp^k \in \caC$ be a sequence generated by the Frank-Wolfe algorithm for~\eqref{eq:minnorm-polyhedron} ($k = 0, 1, \dots$) and $\bmp^*$ be the optimal solution.
    Then, $\norm{\bmp^k}_2^2 - \norm{\bmp^*}_2^2 \leq \frac{4\abs{V}\cdot \norm{b}_2^2}{k+1}$.
    Each iteration of the Frank-Wolfe algorithm takes $O(m + n\log n)$ time.
    %\ynote{should be fixed.}
    %In particular, this yields $O(|V|^2)$ running time per iteration.
\end{theorem}
Theorem~\ref{the:regression-by-Frank-Wolfe} is unsatisfactory because we cannot guarantee $L_F(\bmx) \ni \bmb + \bmp^k$ has a solution for any $k$.
% and hence it is not useful to compute $\bmx \in \bbR^V$ such that $L_F(\bmx)$ has a vector close to $\bmb$.
The algorithm discussed in the next subsection resolves this issue.

\subsection{Combinatorial algorithm}\label{subsec:combinatorial}

In this section, we present a combinatorial algorithm for the minimum norm point problem~\eqref{eq:minnorm-polyhedron}.
First, we show that this problem can be reduced to parametrized submodular minimization.
Although we only need to consider a modular function $H\colon 2^V \to \bbR$, that is, $S \mapsto -\bmb(S)$, to solve~\eqref{eq:minnorm-polyhedron}, we aim to describe it in the most general case.
%Note that the function $F$ here is nothing to do with the submodular transformation considered in~\eqref{eq:minnorm-polyhedron}.
Formally, we need to solve the following:
\begin{align}\label{eq:minnorm}
    \begin{split}
        \max_{\bmw \in P_\caD(H)} & \quad -\frac{1}{2}\norm{\bmw}_2^2.
    \end{split}
\end{align}
Since $P_\caD(H)$ is down-closed, the optimal solution must be a nonpositive vector.
Therefore, we consider the following problem:
\begin{align}\label{eq:regularizedLovasz}
    \begin{split}
        \min_{\bmx \in \bbR^V_+} & \quad h(\bmx) + \frac{1}{2}\norm{\bmx}_2^2, \\
    \end{split}
\end{align}
where $h:\bbR^V \to \bbR$ is the \Lovasz extension of $H$.

\begin{lemma}
    The problems~\eqref{eq:minnorm} and~\eqref{eq:regularizedLovasz} are strong dual to each other.
    Furthermore, if $\bmw^*$ is the optimal solution of~\eqref{eq:minnorm}, then $\bmx^* := - \bmw^*$ is the optimal solution for~\eqref{eq:regularizedLovasz}, and vice versa.
\end{lemma}
\begin{proof}
    Using the Fenchel strong duality, we have
    \begin{align*}
        \min_{\bmx \in \bbR^V_+} \left[ h(\bmx) + \frac{1}{2}\norm{\bmx}_2^2 \right]
        &= \min_{\bmx \in \bbR^V_+}\max_{\bmw \in B_\caD(H)} \left[ \bmw^\top\bmx  + \frac{1}{2}\norm{\bmx}_2^2 \right] \\
        &= \min_{\bmx \in \bbR^V_+}\max_{\bmw \in P_\caD(H)} \left[ \bmw^\top\bmx  + \frac{1}{2}\norm{\bmx}_2^2 \right]
        \tag{since $\bmx \in \bbR^V_+$} \\
        &= \max_{\bmw \in P_\caD(H)} \min_{\bmx \in \bbR^V_+}\left[ \bmw^\top\bmx  + \frac{1}{2}\norm{\bmx}_2^2 \right]
        \tag{by the Fenchel strong duality} \\
        &= \max_{\bmw \in P_\caD(H), \bmw \leq \bmzero} -\frac{1}{2}\norm{\bmw}_2^2.
    \end{align*}
    The second assertion can be checked by direct calculation.
    %\ynote{should be $P_\caD(F)$? should be $B_\caD(F)$? I'm not sure what $\max_{\bmw \in B_\caD(F)}\langle \bmw,\bmx\rangle$ (instead of $\max_{\bmw \in B(F)}\langle \bmw,\bmx\rangle$) is. Is this called the \Lovasz extension of $f:\caD \to \bbR$?}
    %\tnote{$f(\bmx) = \max_{\bmw\in B_\caD(F)} \inprod{\bmw, \bmx}$ is the definition of \Lovasz extension in \cite[Section 6.3]{Fujishige2005}}
\end{proof}

Therefore we focus on~\eqref{eq:regularizedLovasz}.
In what follows, we will show that an optimal solution to~\eqref{eq:regularizedLovasz} can be constructed by solving \emph{parametrized submodular minimization}:
\begin{align}\label{eq:paramSFM}
    \begin{split}
        \min_{S \in \caD} & \quad H(S) + \alpha\abs{S},
    \end{split}
\end{align}
where $\alpha \in \bbR$ is a parameter.
Let $A^\alpha$ be an optimal solution to~\eqref{eq:paramSFM}.

\begin{lemma}[{\citet[Proposition~8.2]{Bach2010}}]\label{lem:bach}
    If $\alpha < \beta$, then $A^\beta \subseteq A^\alpha$.
\end{lemma}
For a vector $\bmx \in \bbR^V$ and a scalar $\alpha \in \bbR$, we write $\set{\bmx \geq \alpha}$ (resp., $\set{\bmx > \alpha}$) to denote the set $\set{v \in V \mid \bmx(v) \geq \alpha}$ (resp., $\set{v \in V \mid \bmx(v) > \alpha}$).
\begin{lemma}\label{lem:subset}
    If $h(\bmx) < +\infty$, then $\set{\bmx \geq \alpha} \in \caD$ for any $\alpha \in \bbR$.
\end{lemma}
\begin{proof}
    Here, we prove the contrapositive.
    Suppose that $\{\bmx \geq \alpha\} \not\in \caD$ for some $\alpha \in \bbR$; this indicates that there exist $u, v \in V$ such that $u \preceq v$ and $\bmx(u) < \bmx(v)$.
    Let us take an arbitrary $\bmw^* \in B(H)$ and consider $\bmw^* + t(\bme_v - \bme_u)$, where $t > 0$.
    Then, since $u \preceq v$, any $S \in \caD$ containing $v$ must also contain $u$.
    Thus $\bmw^* + t(\bme_v - \bme_u) \in B(H)$.
    Since $\bmx(u) < \bmx(v)$, $\langle \bmx, \bmw^* + t(\bme_v - \bme_u)\rangle$ attains infinity as $t \to +\infty$, that is, $h(\bmx) = +\infty$.
\end{proof}

\begin{lemma}
    Define $\bmz \in \bbR_+^V$ as $\bmz(v) := \max\bigset{\sup\{\alpha : v \in A^\alpha\}, 0}\;(v \in V)$.
    Then $\bmz$ is a minimizer of~\eqref{eq:regularizedLovasz}.
\end{lemma}
\begin{proof}
    We can check that $\{\bmz > \alpha\} \subseteq A^\alpha \subseteq \{\bmz \geq \alpha\}$ for $\alpha \geq 0$ using the definition of $\bmz$ and Lemma~\ref{lem:bach}. %\ynote{are we assuming $\alpha \geq 0$ (otherwise it implies $V \subseteq A^\alpha$)? Also, how did we use Lemma~\ref{lem:subset}?}
    Therefore, $A^\alpha = \{\bmz \geq \alpha\}$ almost everywhere.
    Let $\bmx \in \bbR_+^V$ be a feasible solution such that $h(\bmx) < +\infty$.
    Then for any $\alpha \in \bbR$, $\{\bmx \geq \alpha\} \in \caD$ holds through Lemma~\ref{lem:subset}.
    Furthermore, we have
    \begin{align*}
        h(\bmz) + \frac{1}{2}\norm{\bmz}_2^2
        &= \int_0^\infty H(\{\bmz \geq \alpha\}) d\alpha + \sum_{v \in V} \int_0^{\bmz(v)} \alpha d\alpha \\
        &= \int_0^\infty \left[ H(\{\bmz \geq \alpha\}) d\alpha + \sum_{v \in V} \alpha\bmone_{\alpha \geq \bmz(v)} \right] d\alpha \\
        &\leq \int_0^\infty \left[ H(\{\bmx \geq \alpha\}) d\alpha + \sum_{v \in V} \alpha\bmone_{\alpha \geq \bmx(v)} \right] d\alpha \\
        &= h(\bmx) + \frac{1}{2}\norm{\bmx}_2^2,
    \end{align*}
    where the inequality follows since the integrand is equal to $H(A^\alpha) + \alpha\abs{A^\alpha}$ almost everywhere.
\end{proof}

Now, we consider the following modular case: $H(S) = - \sum_{i \in S} \bmb(i)$.
Considering the abovementioned arguments, we only need to solve
$\min_{S \in \caD} \sum_{i \in S} (\alpha - \bmb(i))$ for all $\alpha \in \bbR$.
Using the approach used in~\cite{Picard1982}, we define the following directed graph.
Let $G' = (U, A)$ be the digraph corresponding to the Birkoff representation of $\caD$.
Then, define a directed graph $G = (U \cup \{s, t\}, A \cup \bar{A})$, where
$\bar{A} := \{ su : u \in U\} \cup \{ ut: u \in U\}$.
In addition, define a capacity function $c$ on $\in A \cup \bar{A}$ as
\begin{align}\label{eq:capacity}
    c(a) :=
    \begin{cases}
        +\infty & \text{if $a \in A$,} \\
        \min\{- \alpha + \bmb(u), 0\} & \text{if $a = su$,} \\
        \min\{\alpha - \bmb(u), 0\} & \text{if $a = ut$}.
    \end{cases}
\end{align}
Then, the minimum $st$-cuts in $G$ provide the desired minimizers.
Refer to Figure~\ref{fig:modular-optimization} for an illustrative example.

\begin{figure}[t!]
\centering
\begin{minipage}[b]{0.45\linewidth}
\centering
\begin{tikzpicture}[line width=1pt]
	\tikzstyle{nodestyle} = [circle, draw, inner sep=4pt, minimum height=25pt];
	\tikzstyle{edgestyle} = [-{Triangle}];
	\node[nodestyle] (v4) {$4$};
	\node[nodestyle] (v12) at ([xshift=-30pt, yshift=50pt]v4) {$12$};
	\node[nodestyle] (v3)  at ([xshift=+30pt, yshift=50pt]v4) {$3$};
	\draw[edgestyle] (v12) -- (v4);
	\draw[edgestyle] (v3)  -- (v4);
\end{tikzpicture}
\subcaption{The Birkoff representation of $\caD$.\label{fig:Birkoff}}
\end{minipage}
\hspace{1em}
\begin{minipage}[b]{0.45\linewidth}
\centering
\begin{tikzpicture}[line width=1pt]
	\tikzstyle{nodestyle} = [circle, draw, inner sep=4pt];
	\tikzstyle{edgestyle} = [-{Triangle}];
	\node[nodestyle, minimum height=25pt] (v4) {$4$};
	\node[nodestyle, minimum height=25pt] (v12) at ([xshift=-30pt, yshift= 50pt]v4) {$12$};
	\node[nodestyle, minimum height=25pt] (v3)  at ([xshift=+30pt, yshift= 50pt]v4) {$3$};
	\node[nodestyle, minimum height=15pt] (s)   at ([xshift=-50pt, yshift=  0pt]v12) {$s$};
	\node[nodestyle, minimum height=15pt] (t)   at ([xshift=+50pt, yshift=-30pt]v4) {$t$};

	\draw[edgestyle] (v12) to node[left] {$+\infty$} (v4);
	\draw[edgestyle] (v3)  to node[right] {$+\infty$} (v4);
    \draw[edgestyle] (s)  to [bend right = 50] node[below, xshift=-20pt] {\small $-\alpha + \bmb(12)$} (v12);
    \draw[edgestyle] (s)  to [bend left = 50]  node[above] {\small $-\alpha + \bmb(3)$} (v3);
    \draw[edgestyle] (v4) to [bend right = 20] node[right, pos=0.1] {\small $\alpha - \bmb(4)$} (t);
	\path[opacity=0.2, fill] ([xshift=10pt,yshift=10pt]v12.north east) -- ([xshift=10pt,yshift=10pt]v4.north east) -- ([xshift=10pt,yshift=-10pt]v4.south east) -- ([xshift=-10pt,yshift=-10pt]v4.south west) -- ([xshift=-10pt,yshift=-10pt]s.south west) to node[yshift=18pt,opacity=0.6]{$S \cup \{s\}$} ([xshift=-10pt,yshift=10pt]s.north west) -- ([yshift=10pt]v12.north west);
\end{tikzpicture}
\subcaption{The directed graph $G$ and its minimum $st$-cut $S \cup \set{s}$. The arcs with zero capacity are ommited.\label{fig:st-cut}}
\end{minipage}
\caption{This figure illustrates the method used to solve $\min_{S \in \caD} \sum_{i \in S} (\alpha - \bmb(i))$ using minimum $st$-cut for $\caD = \{\emptyset, \{4\}, \{3,4\}, \{1,2,4\}, \{1,2,3,4\} \}$.\label{fig:modular-optimization}}
\end{figure}
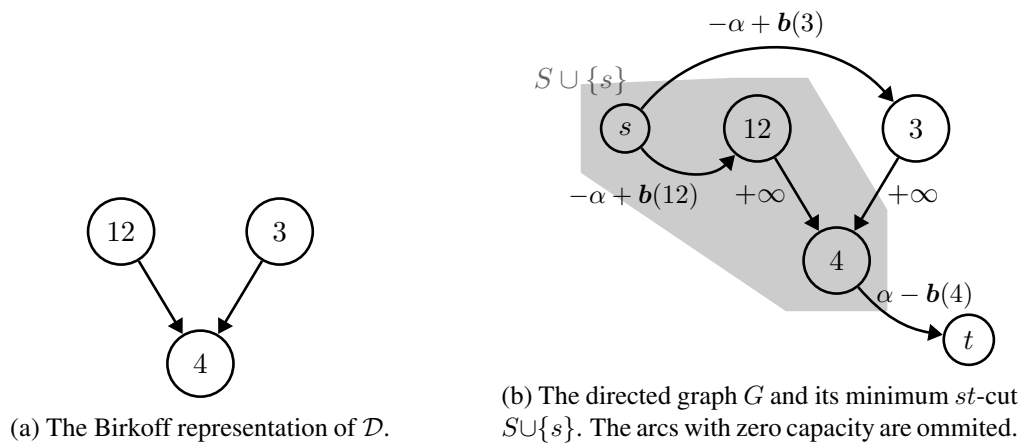

Furthermore, the capacity function $c$ satisfies the so-called \emph{GGT} structure~\citep{Gallo1989}\footnote{The original GGT structure assumes that the capacities of arcs from source $s$ are nondecreasing, whereas the capacities of arcs to sink $t$ are nonincreasing, while the others are constant. We can transform our capacity function to this setting by replacing $\alpha$ with $-\alpha$.};
the capacities of arcs from source $s$ are nonincreasing, whereas the capacities of arcs to sink $t$ are nondecreasing, while the others are constant.
For such a capacity function, all the values of $\alpha$ at which the value of minimum $st$-cuts change can be computed in $O(nm^2 \log\frac{n^2}{m})$ time using a parametric flow algorithm~\citep{Gallo1989}, where $n$ and $m$ are the number of vertices and arcs in $G$, respectively.

\begin{theorem}
    Assume that the Birkoff representation of $\caD$ is given as a directed graph with $n$ vertices and $m$ arcs.
    Then, there exists a strongly polynomial time algorithm for~\eqref{eq:minnorm-polyhedron} with time complexity $O(nm^2 \log\frac{n^2}{m})$.
    %In particular, this leads to $O(|V|^5)$ running time.
\end{theorem}
Theorem~\ref{the:intro-regression} is immediate from the above theorem and the fact that we can construct the Birkoff representation of $\ker(F)$ in polynomial time as previously discussed.

\bibliography{main}

\appendix
%!TEX root=./main.tex

\section{Semi-supervised Learning}\label{sec:semi-supervised}
In this section, we prove Theorem~\ref{the:intro-semi-supervised}.
%extend Theorem~\ref{the:intro-system} to the setting where some part of $\bmx$ is given as the input. This extension makes us possible to apply our framework to semi-supervised learning.

We consider the problem of solving $\bmb \in L_F(\bmx)$ under constraints that $\bmx(v) = \tilde{\bmx}(v)$ for all $v \in T$ and $\bm\bmb(v) = \tilde{\bmb}(v)$ for all $v \in U$. To find such $\bmx$ and $\bmb$, we consider the following optimization problem.
\begin{align}
	\begin{array}{lll}
		\label{eq:semi-supervised-primal}
		\displaystyle \min_{\bmx \in \bbR^V,\bmeta \in \bbR^E} & \displaystyle \frac{1}{2} \| \bmeta \|^2 - \sum_{v \in U} \tilde{\bmb}(v) \bmx(v)& \\
		\text{subject to} & \displaystyle \bmx(v) = \tilde{\bmx}(v) & (v \in T)\\
		 & \displaystyle f_e(\bmx) \le \bmeta(e) & (e \in E)
	\end{array}
\end{align}
The following theorem shows that~\eqref{eq:semi-supervised-primal} can be solved in polynomial time, and an obtained optimal solution $\bmx$ satisfies $\bmb \in L_F(\bmx)$.
\begin{theorem}
	There is a polynomial time algorithm that solves \eqref{eq:semi-supervised-primal}.
	In addition, for an optimal solution $(\bmx^*,\bmeta^*)$ to~\eqref{eq:semi-supervised-primal}, $\bmx^*(v) = \tilde{\bmx}(v)$ holds for each $v \in T$, and there exists some $\bmb \in L(\bmx^*)$ such that $\bm\bmb(v) = \tilde{\bmb}(v)$ holds for each $v \in U$.
\end{theorem}
\begin{proof}
	\eqref{eq:semi-supervised-primal} is equivalent to the following problem:
	\begin{align}
	\begin{array}{lll}
		\displaystyle \min_{\bmx \in \bbR^V} & \displaystyle \frac{1}{2} \sum_{e \in E} f_e(\bmx)^2  - \sum_{v \in U} \tilde{\bmb}(v) \bmx(v)& \\
		\text{subject to} & \displaystyle \bmx(v) = \tilde{\bmx}(v) & (v \in T)\\
	\end{array}
	\end{align}
	This problem can be viewed as an unconstraint convex programming with variables $(\bmx(v))_{v \in U}$. As we can compute the subdifferential of this objective function, we can solve~\eqref{eq:semi-supervised-primal} with the ellipsoid method in polynomial time.

	Let $(\bmx^*,\bmeta^*)$ be an optimal solution. Since $\bmx^*$ is a solution to~\eqref{eq:semi-supervised-primal}, $\bmx^*(v) = \tilde{\bmx}(v)$ holds for each $v \in T$. From the first-order optimality condition, there exists $\bmw_e \in \partial f_e(\bmx^*)$ for each $e \in E$ such that $\sum_{e \in E} \bmw_e(v) f_e(\bmx^*) - \tilde{\bmb}(v) = 0$ for all $v \in U$. It follows that $\bmb \in L_G(\bmx^*)$ for some $\bmb$ such that $\bm\bmb(v) = \tilde{\bmb}(v)$.
\end{proof}

To derive more efficient algorithms for special cases such as directed graphs and hypergraphs, we introduce a flow-like formulation, which is an extension of~\eqref{eq:flowlike-primal} and~\eqref{eq:flowlike-dual}.

Let $\caV_e$ be the set of extreme points of $B(F_e)$ for each $e \in E$.
Then, the original primal problem~\eqref{eq:semi-supervised-primal} is equivalent to:
\begin{align}
	\begin{array}{lll}
		\label{eq:semi-supervised-flowlike-primal}
		\displaystyle \min_{\bmx \in \bbR^V,\bmeta \in \bbR^E} & \displaystyle \frac{1}{2} \| \bmeta \|^2 - \sum_{v \in U} \tilde{\bmb}(v) \bmx(v)& \\
		\text{subject to} & \displaystyle \sum_{v \in U} \bmw(v) \bmx(v) + \sum_{v \in T} \bmw(v) \tilde{\bmx}(v) \le \bmeta(e) & (\bmw \in \caV_e, e \in E)
	\end{array}
\end{align}
with regarding $(\bmx(v))_{v \in U}$ as variables.
The dual problem is
\begin{align}
	\begin{array}{lll}
		\label{eq:semi-supervised-flowlike-dual}
		\displaystyle \min_{\bmphi \geq \bmzero} & \displaystyle \frac{1}{2} \sum_{e \in E} \left( \sum_{\bmw \in \caV_e} \bmphi(e, \bmw) \right)^2 - \sum_{e \in E, \bmw \in \caV_e} \bmphi(e, \bmw) \left( \sum_{v \in T} \bmw(v) \tilde{\bmx}(v) \right)& \\
		\text{subject to} & \displaystyle \sum_{e \in E, \bmw \in \caV_e} \bmphi(e, \bmw) \bmw(v) = \tilde{\bmb}(v) & (v \in U).\\
	\end{array}
\end{align}

The first constraint can be interpreted as a flow boundary constraint on $v \in U$.
In the cases of constant arity submodular transformations, including directed graphs and hypergraphs, the number of variables $(\bmphi(e, w))_{e \in E, w \in \caV_e}$ in the dual problem is bounded by a polynomial in $|V|$.
As discussed in Section~\ref{subsec:system-constant-arity}, if $F$ is given by a directed graph or a hypergraph, this dual problem can be reduced to the quadratic cost flow problem, which  can be solved in $O(|E|^4 \log |E|)$ time~\citep{Vegh:2012jm}.

%with the submodular Laplacian systems, the Frank-Wolfe algorithm can be applied.
%If $(F_e)_{e \in E}$ is directed graphs, this dual problem can be reduced to the quadratic cost flow problem and solved in $O(|E|^4 \log |E|)$ time~\citep{Vegh:2012jm}.

%!TEX root=./main.tex

\section{Triangle Inequality of Effective Resistance}\label{sec:effective-resistance}
    %\begin{align}
	%	\min_{\bmx,\bmeta} \frac{1}{2}\| \bmeta \|_2^2  - \langle \bmb, \bmx \rangle \text{ subject to } f_e(\bmx) \leq \eta(e) \quad (e \in E).
    %\end{align}
    %The corresponding Lagrangian is
    %\begin{align*}
    %    P(\bmx,\bmeta,\bmphi) &= \frac{1}{2}\|\bmeta\|_2^2 -\langle \bmb,\bmx\rangle - \sum_{e \in E} \bmphi(e) (f_e(\bmx) - \bmeta(e)) \\
    %                        &= \frac{1}{2}\|\bmeta\|_2^2 +\langle \bmphi, \bmeta \rangle - \left[ \langle \bmb, \bmx \rangle - \sum_{e \in E} \bmphi(e) f_e(\bmx) \right].
    %\end{align*}
    %Therefore we obtain the dual problem:
    %\begin{align}
	%	\max_{\bmphi} - \frac{1}{2}\| \bmphi \|_2^2 \text{ subject to } \sum_{e \in E} \bmphi(e) f_e(\bmx) - \langle \bmb, \bmx \rangle \ge 0 \text{ for all }\bmx \in \bbR^V.
    %\end{align}

% \begin{definition}{(Effective resistance)}
%   Let $u, v \in V$ be elements of $V$ and suppose~\eqref{eq:general-energy-minimization} is feasible for $\bmb = \bme_u - \bme_v$. The effective resistance $R_{uv}$ from $u$ to $v$ is defined to be the double of the optimal value of~\eqref{eq:general-energy-minimization}, i.e., $R_{uv} = \| \bmphi \|^2$ where $\bmphi$ is an optimal solution for~\eqref{eq:general-energy-minimization} in the case of $\bmb = \bme_u - \bme_v$.
% \end{definition}

%  For an undirected graph, the effective resistance from $u \in V$ to $v \in V$ is defined as $(\bme_u - \bme_v) L_F^{+} (\bme_u - \bme_v)$ in previous work. The following lemma shows our definition of the effective resistance is equivalent to this previous definition in the case of undirected graphs.
In this section, we prove Theorem~\ref{the:intro-triangle-inequality}.

We start with rephrasing effective resistance.
\begin{lemma}\label{lem:effective-resistance-rephrase}
  Let $F\colon 2^V \to \bbR_+^E$ be a submodular transformation.
  Suppose~\eqref{eq:general-energy-minimization} is feasible for $\bmb = \bme_u - \bme_v$, and let $\bmphi^* \in \bbR^E$ be an optimal solution.
  Then we have $R_F(u,v) = \| \bmphi^* \|^2$.
\end{lemma}
\begin{proof}
  Let $\bmx^* \in \bbR^V$ be an optimal solution to the primal problem~\eqref{eq:general-primal} for $\bmb = \bme_u - \bme_v$.
  From the strong duality, we have $$\frac{1}{2} \sum_{e \in E} f_e(\bmx^*)^2 - \langle \bme_u - \bme_v, \bmx^* \rangle = - \frac{1}{2} \| \bmphi^* \|^2.$$
  From the complementary slackness condition, we have $\varphi^*(e) = f_e(\bmx^*)$ for all $e \in E$.
  In addition, due to the first-order optimality condition, we have $L_F(\bmx^*) \ni \bme_u - \bme_v$.
  %Since $L_F$ is an undirected graph Laplacian and the primal is bounded, $\bmx^* = L_F^{+} (\bme_u - \bme_v)$.
  Substituting them to the above equation, we obtain
  \[
    \| \bmphi^* \|^2 = (\bme_u - \bme_v) L_F^{+} (\bme_u - \bme_v),
  \]
  and the claim holds.
\end{proof}

\begin{lemma}{(Triangle inequality of effective resistance)}
	Let $F\colon 2^E \to \bbR_+^E$ be a submodular transformation.
  Suppose that effective resistances $R_F(u,v)$ and $R_F(v,w)$ are bounded.
  Then, it holds that $$R_F(u,v) + R_F(v,w) \ge R_F(u,w).$$
\end{lemma}
\begin{proof}
	Suppose $(\bmx_{uv}, \bmphi_{uv})$, $(\bmx_{vw}, \bmphi_{vw})$ and $(\bmx_{uw}, \bmphi_{uw})$ are saddle points of the above Lagrangian~\eqref{eq:general-lagrangian} for $\bmb = \bme_u - \bme_v$, $\bmb = \bme_v - \bme_w$ and $\bmb = \bme_u - \bme_w$, respectively.
  Let $\bmphi_{uv} \vee \bmphi_{vw} \in \bbR^E$ be $\bmphi_{uv} \vee \bmphi_{vw}(e) = \max\set{\bmphi_{uv}(e), \bmphi_{vw}(e) }$.
  We show $\bmphi_{uv} \vee \bmphi_{vw}$ is a feasible solution to~\eqref{eq:general-energy-minimization} for $\bmb = \bme_u - \bme_w$.
  It is enough to show for any $S \subseteq V$, the constraint $$\sum_{e \in E} (\bmphi_{uv} \vee \bmphi_{vw})(e) F_e(S) - (\bme_u - \bme_w)(S) \ge 0$$ holds. If $u \not\in S$ or $w \in S$, the constraint holds since $(\bme_u - \bme_w)(S) \le 0$ and $F_e$ is non-negative for each $e \in E$. Hence we consider $S$ such that $u \in S$ and $w \not\in S$.
  When $v \in S$, we have
	\begin{align*}
		&\sum_{e \in E} (\bmphi_{uv} \vee \bmphi_{vw})(e) F_e(S) - (\bme_u - \bme_w)(S)
		\ge \sum_{e \in E} \bmphi_{vw}(e) F_e(S) - 1\\
		&= \sum_{e \in E} \bmphi_{vw}(e) F_e(S) - (\bme_v - \bme_w)(S)
		\ge 0.
	\end{align*}
	The last inequality holds because $\bmphi_{vw}$ is a feasible solution to~\eqref{eq:general-energy-minimization} for $\bmb = \bme_v - \bme_w$.

	Similarly, when $v \not\in S$, we have
	\begin{align*}
		&\sum_{e \in E} (\bmphi_{uv} \vee \bmphi_{vw})(e) F_e(S) - (\bme_u - \bme_w)(S)
		\ge \sum_{e \in E} \bmphi_{uv}(e) F_e(S) - 1\\
		&= \sum_{e \in E} \bmphi_{uv}(e) F_e(S) - (\bme_u - \bme_v)(S)
		\ge 0.
	\end{align*}

	Since $\bmphi_{uv} \vee \bmphi_{vw}$ is a feasible solution to~\eqref{eq:general-energy-minimization} for $\bmb = \bme_u - \bme_w$, by Lemma~\ref{lem:effective-resistance-rephrase}, we can show the triangle inequality as follows:
	\begin{align*}
		R_F(u,w) &= \|\bmphi_{uw}\|^2
		\le \|\bmphi_{uv} \vee \bmphi_{vw}\|^2
		\le \|\bmphi_{uv}\|^2 + \|\bmphi_{vw}\|^2
		= R_F(u,v) + R_F(v,w),
	\end{align*}
  and the claim holds.
\end{proof}

\end{document}